\documentclass[a4paper,12pt]{article}

\usepackage[a4paper,left=2.5cm,right=2.5cm,top=2.5cm,bottom=3.0cm]{geometry}

\usepackage{graphicx}
\usepackage{graphics}
\usepackage{color}
\usepackage{amsmath}
\usepackage{amssymb}
\usepackage{amsfonts}
\usepackage[english]{babel}
\usepackage{enumitem}
\usepackage{comment}
\usepackage[nottoc]{tocbibind}
\usepackage{amsthm}
\usepackage{mathtools}
\usepackage{caption}
\usepackage{subcaption}
\usepackage{subfig}
\DeclarePairedDelimiter{\ceil}{\lceil}{\rceil}

\usepackage{algpseudocode}
\usepackage[nothing]{algorithm}
\algtext*{EndWhile}
\algtext*{EndIf}
\algtext*{EndFor}

\floatname{algorithm}{Procedure}

\newcommand{\ProcCS}{\texttt{\textup{GridSearching}}}
\newcommand{\ModProcCS}{\texttt{\textup{ModGridSearching}}}

\newcommand{\ProcClean}{\texttt{\textup{CleanExpansion}}}
\newcommand{\ProcUpgrade}{\texttt{\textup{UpgradeCheckpoints}}}
\newcommand{\ProcConnSearch}{\texttt{\textup{ConnectedSearching}}}
\newcommand{\ModProcConnSearch}{\texttt{\textup{ModConnectedSearching}}}
\newcommand{\mcs}[1]{\texttt{mcs}(#1)}
\newcommand{\csn}[1]{\texttt{cs}(#1)}

\newcommand{\fr}[2]{F\left(\left(#1\right),\left(#2\right)\right)}
\newcommand{\aFr}{{\cal F}}
\newcommand{\aC}{{\cal C}}
\newcommand{\cS}{{\cal S}}

\newcommand{\size}{\sqrt{n}}
\newcommand{\border}[1]{\delta(#1)}

\newcommand{\grid}{G}

\newcommand{\ver}[2]{v\left(#1,#2\right)}

\newcommand{\nrPred}{10}
\newcommand{\checkpoints}{\mathcal{C}}

\usepackage{epstopdf}

\newtheorem{theorem}{Theorem}[section]
\newtheorem{lemma}{Lemma}[section]

\newtheorem{observation}{Observation}[section]

\newtheorem{remark}{Remark}[section]

\ifpdf
\fi

\newcommand{\weightOfC}[2]{\omega_{#2}(#1)}
\newcommand{\expansionOfC}[2]{\mathcal{E}(#1,#2)}
\newcommand{\bottleneckOfC}[1]{b(#1)}
\newcommand{\leqF}{\prec}
\newcommand{\rectOfF}[2]{\mathcal{R}(#1,#2)} 
\newcommand{\expOfC}[2]{#1\langle #2\rangle}
\newcommand{\expOfCplus}[2]{#1^{+}\langle #2\rangle}

\newcommand{\st}{\hspace{0.1cm}\bigl|\bigr.\hspace{0.1cm}}

\begin{document}

\renewcommand{\thefootnote}{\fnsymbol{footnote}}
\footnotetext[1]{Research partially supported by National Science Centre, Poland, grant number 2015/17/B/ST6/01887.}
\footnotetext[2]{Faculty of Electronics, Telecommunications and Informatics, Gda{\'n}sk University of Technology, Gda{\'n}sk, Poland}

\title{Distributed Searching of Partial Grids\footnotemark[1]}
\author{Dariusz Dereniowski\footnotemark[2] \and
        Dorota Urba{\'n}ska\footnotemark[2]}

\date{}

\maketitle
\thispagestyle{empty}

\floatname{algorithm}{Procedure}

\begin{abstract}
We consider the following distributed pursuit-evasion problem. A team of mobile agents called searchers starts at an arbitrary node of an unknown $n$-node network. 
Their goal is to execute a search strategy that guarantees capturing a fast and invisible intruder regardless of its movements using as few agents as possible.
We restrict our attention to networks that are embedded into partial grids: nodes are placed on the plane at integer coordinates and only nodes at distance one can be adjacent.
We give a distributed algorithm for the searchers that allow them to compute a connected and monotone strategy that guarantees searching any unknown partial grid with the use of $O(\size)$ searchers.
As for a lower bound, not only there exist partial grids that require $\Omega(\size)$ searchers, but we prove that for each distributed searching algorithm there is a partial grid that forces the algorithm to use $\Omega(\size)$ searchers but $O(\log n)$ searchers are sufficient in the offline scenario.
This gives a lower bound of $\Omega(\size/\log n)$ in terms of achievable competitive ratio of any distributed algorithm.
\end{abstract}

\vspace*{1cm}

{\bf keywords:} connected search number, distributed searching, graph searching, partial grid, pursuit-evasion

\section{Introduction}
A team of mobile autonomous robots wants to search an area with the goal of finding a mobile intruder (or lost entity).
The intruder has several properties that dictate how a search should be conducted.
First, the intruder is invisible and therefore the robots may conclude its potential locations only from the history of their own moves.
Second, it is assumed that the speed of the intruder is unknown and therefore the robots build their search strategy assuming that the intruder is very fast: may traverse arbitrarily long distance between any two actions of a robot.
Third, the intruder is very clever, i.e., it will avoid being captured as long as possible; in other words we may imagine that it knows locations of robots and their future movements at any point.
This assumption enforces robots to consider the worst case scenario for them since they want to have a search strategy that guarantees interception.
The above problem is usually restated in discrete terms, naturally expressing the search game using graph-theoretic notation.
Following the widely used terminology, the mobile entities performing the search are called \emph{searchers}.

\medskip
In this work we focus on the graph-theoretic problem statement, where the searchers operate in a given graph in which they move along edges.
Moreover, what greatly influences algorithmic approach is assumption whether the searchers know the graph in advance (offline version of the problem) or whether the graph is unknown and the searchers learn its structure while conducting the search (online or distributed setting).
We will shortly review both approaches, giving later a formal statement of the problem we study in this work.
In all cases we are interested in minimizing the number of searchers needed to clear the given network.\footnote{In this work the terms graph and network are used exchangeably.}

\medskip\noindent
\textbf{Off-line searching}.
The offline graph searching models are extensively studied and numerous deep results have been obtained, providing insight into not only the problem itself but also enriching the more widely understood graph theory through the connections between graph searching games and many graph parameters, e.g., pathwidth, treewidth, branchwidth, bandwidth, profile, interval thickness, vertex separation number; see e.g. \cite{FT08} for a survey and further references.
The historically first studied graph searching model is called \emph{edge search} \cite{Parsons76,Petrov82}.
In this problem, the goal is to construct a search strategy that guarantees capturing a fast and invisible fugitive (thus, the strategy must ensure success regardless of the moves performed by the fugitive) in a graph that is given as an input to an algorithm computing a search strategy.
A search strategy itself is a sequence of \emph{moves}, where each move is one of the following: (i) placing a searcher on any graph node; (ii) removing a searcher from the node it occupies; (iii) sliding a searcher along an edge in order to clear it.
Since we adopt the connected searching problem in our distributed model, we point out to few recent works on the problem
\cite{BFFFNST12,BestGTZ2015,Dereniowski11,Dereniowski12,FlocchiniHL07,FlocchiniHL08}.

\medskip\noindent
\textbf{Distributed searching}.
In the distributed, or on-line, version of the problem it is assumed that the network is unknown in advance to the searchers.
In this setting, some assumptions need to be made.
First, only \emph{monotone} search strategies are considered.
In a monotone search strategy, once an edge has been cleared, it must remain clean till the end of the search; in other words, the subgraph composed of edges that may contain the fugitive may only shrink as the search progresses thus disallowing any recontaminations.
This assumption is dictated by an observation that otherwise the searchers may first learn the structure of the network by exploring it (and thus ignoring the possibility of capturing the intruder at this stage) and once the network is known, they can compute a search strategy by using an off-line algorithm and finally execute the strategy. The problem then reduces to exploration and map construction, well studied problems in distributed computing.
Another natural assumption is to forbid placing a searcher on a node that has not been visited before.

We consider \emph{connected} search strategies in this work, i.e., strategies that guarantee that at any given time point the subgraph that is clean is connected.
Note that this allows us to assume that all searchers start at some node called the \emph{homebase} and only moves of type (iii) are then made (see the definition of edge search above).
Indeed, removing a searcher from a node $u$ and placing it on another $v$ (i.e., \emph{jumping}) one may be replaced\footnote{Note that this observation is true as long as the considered optimization criterion is to minimize the number of searchers. For other criteria, like e.g. search time, the exclusion of jumping may be potentially limiting.} by a sequence of sliding moves along a path from $u$ to $v$ consisting of clean edges only (such a path must exists due to connectedness and monotonicity).

\subsection{Related work}

\noindent\textbf{Off-line problems}.
One of the central questions raised in the context of various graph searching problems is the monotonicity which can be stated as follows.
A searching problem is said to be \emph{monotone} if there exists a monotone search strategy solving the problem.
Note that proving monotonicity is a tool that allows to conclude membership in NP for a given problem.
It is known that the edge search problem is monotone \cite{LaPaugh93}.
On the other hand, the connected search is not monotone \cite{YDA09}.
A question related to the latter searching model is: how many extra searchers one needs to ensure connectivity.
It turns out that each monotone edge search strategy can be converted (in polynomial time) into a monotone connected one by approximately `doubling' the number of searchers \cite{Dereniowski12SIDMA}.
Thus, for asymptotic results, like the one in this work, this gives another reason that justifies restricting attention to monotone connected search strategies.

\medskip\noindent\textbf{Distributed searching}.
In most cases, when designing distributed searching algorithms, the monotonicity requirement is adopted.
(See \cite{BlinFNV08} for an example how an optimal connected search strategy can be constructed in a distributed fashion when recontamination is allowed.)
During construction of a monotone strategy in a distributed way, there is naturally some `cost' involved in terms of increased number of searchers required for guarding --- this cost measured as the ratio of number of searchers that each distributed algorithm needs to use for some $n$-node graph and its monotone connected (off-line) search number is know to be $\Omega(n/\log n)$ \cite{INS09}.
In the realm of distributed algorithms, natural questions arise with respect to the amount (and type) of additional information regarding the underlying network given a priori to an algorithm.
In \cite{NS09} is was proved that $O(n\log n)$ bits of \emph{advice} are sufficient to construct an optimal connected monotone search strategy (the concept of such quantitative approach to advice analysis was introduced in \cite{FIP06}).
An example of an algorithmic approaches in a very weak computational model see e.g. \cite{BlinBN12,DAngeloNN14}.

Grid networks were studied in \cite{DaadaaFZ10} where the searching model used the concept of temporal immunity: a node after cleaning remains protected (even if unguarded) against recontamination for a certain amount of time.
For other searching works involving immunity see e.g. \cite{FlocchiniLPS13,FlocchiniLPS16}.
For other distributed searching models and algorithms for specific network topologies see \cite{CaiFS14,FlocchiniHL08,GoncalvesLMF10,MoravejiSZ11}.
We provide in Section~\ref{sec:conclusions} a brief discussion of a potential applicability of our result in `continuous' environments, like polygons.

\medskip\noindent\textbf{Applications in robotics}.
We note that our results may be of particular interest not only providing theoretical insight into searching dynamics in distributed agent computations, but may also find applications in the field of robotics.
Most investigations oriented towards algorithms that can be applied on physical devices need to deal with the problem of modeling of real world.
This can be done either by discretizing it (usually via graph theory) or by building algorithms that work in continuous search space and need to address the geometric issues that emerge.
(In Section~\ref{sec:polygon}, we add a brief discussion on this subject from the point of view of our results.)
Having in mind vast literature on the subject we point out interested reader by providing few references to recent works in this field \cite{ChungHI11,DurhamFB12,HSDK09,KC10,RaboinKN12,RobinL16,RodriguezDBMMMKZA11,SachsLR04,StifflerO14}.

\subsection{Outline of this work}
The next section introduces the notation used in this work and provides problem statement.
It is subdivided so that Section~\ref{sec:search-def} defines the graph searching problem we study while Section~\ref{sec:gridnotation} introduces the terminology related to the partial grid networks we consider in this paper.
Section~\ref{sec:lower} provides a construction of a class of $n$-node networks such that each distributed algorithm uses $\Theta(\sqrt{n})$ searchers which turns out to be $\Omega(\sqrt{n}/\log n)$ times more that an optimal off-line algorithm would use (recall that by off-line algorithm we refer to the case when the entire network is given as an input and hence is known in advance to the algorithm).
This serves as a lower bound in our analysis.

Section~\ref{sec:algorithm} describes a distributed algorithm that performs a guaranteed search in partial grids and it is assumed that the algorithm is given an upper bound $n$ on the size of the network.
We point out that this algorithm uses a distributed procedure from \cite{BorowieckiDK15} as a subroutine that is called many times to clear selected parts of a grid, and it can be seen as a generalization from a `linear' graph structure studied in \cite{BorowieckiDK15} to a $2$-dimensional structure discussed in this work.
Also, although both algorithms are conducted via some greedy rules which dictate how a search should `expand' to unknown parts of the graph, the analysis of our algorithm is different from the one in \cite{BorowieckiDK15}.

Then, in Section~\ref{sec:proof} we prove the correctness of the algorithm and provide an upper bound on its performance: it is using $O(\sqrt{n})$ searchers for any partial grid network.
In Section~\ref{sec:unknown} we consider a modified version of the algorithm, which receives no information on the underlying graph in advance, and we prove that the algorithm also uses $O(\sqrt{n})$ searchers.
This result, stated in Theorem~\ref{thm:unknown}, is the main contribution of this work.
We finish with conclusions in Section~\ref{sec:conclusions}, giving few remarks on how our work relates to searching two-dimensional environments, like polygons with holes.
As there are many open problems and research directions related to the subject, we list some of them also in Section~\ref{sec:conclusions}.

\section {Definitions and terminology} \label{sec:terminology}
In this section we will present the notation we use.
We consider only simple undirected connected graphs.
Given any graph $G=(V,E)$ and $X\subseteq V$, $\grid[X]$ is the subgraph of $G$ \emph{induced by} $X$: its node set is $X$ and consists of all edges $\{u,v\}$ of $G$ having both endpoints in $X$.

\subsection{Problem statement} \label{sec:search-def}

A \emph{connected $k$-search strategy} $\cS$ for a network $G$ is defined as follows.
Initially, $k$ searchers are placed on a node $h$ of $G$, called the \emph{homebase}.
(We also say that $\cS$ \emph{starts at} $h$.)
Then, $\cS$ is a sequence of moves, where each \emph{move} consists of selecting one searcher present at some node $u$ and sliding the searcher along an edge $\{u,v\}$.
(Thus, the searcher moves from its current location to one of the neighbors.)
Initially, all edges are \emph{contaminated}.

After each move of sliding a searcher along an edge $\{u,v\}$ it is declared to be \emph{clean}. It becomes contaminated again (\emph{recontaminated}) if at any time during execution of the strategy $\cS$ at least one of its endpoints is not occupied by a searcher and is incident to a contaminated edge.
If no recontamination happens in $\cS$, then $\cS$ is called \emph{monotone}.
Regardless of whether the strategy is monotone or not, we require that the subgraph consisting of all clean edges is connected after each move of the search strategy.
Finally, we require that after the last move of $\cS$ all edges are clean.

The minimum $k$ such that there exist a node $h$ and a (monotone) connected $k$-search strategy that starts at $h$ is called the (monotone) \emph{connected search number} of $G$ and denoted by ($\mcs{G}$, respectively) $\csn{G}$.

\medskip
Having defined a search strategy, we now state the distributed model we use.
All searchers start at the homebase --- a selected node of the network.
The network itself is not known in advance to the searchers.
In fact, the searchers have no information about the network.
(We note here that our main algorithmic result will be obtained in two stages: first we describe an algorithm that as an input receives the upper bound $n$ of the size of the network and then we use it to obtain our main result, an algorithm that works without any a priori information about the network.)
We assume that nodes are anonymous and searchers have identifiers.
The edges incident to each node are marked with unique labels (port numbers) and because only partial grids are considered in this work (for a definition see Section~\ref{sec:gridnotation}) we assume that labels naturally reflect all possible directions for each edge (i.e., left, right, up and down).

For the searchers, we assume that they communicate \emph{locally} by exchanging information when present at the same node.
Our algorithm is stated as there existed \emph{global} communication but it can be easily turned into required one with local communication as follows: we can designate one searcher called the \emph{leader} who will be performing the following actions at the beginning of each move of the search strategy to be executed.
First, the leader visits all nodes of the subgraph searched to date and gathers complete information about its structure and positions of all other searchers, then the leader computes the next move and finally visits all searchers to pass the information about the next move.
Then, the move is performed by the agents.\footnote{Note that the actions of the leader clearly contain a lot of excess work in terms of the number of moves it performs; since the criteria as time or cost (number of sliding moves) are out to scope of this work, we will leave the reader with such a simple leader implementation.}

Our algorithm is described for the synchronous model in which time is divided into steps, each step having the same unit length duration allowing each searcher to perform its local computations and slide along an edge if the searcher decides to move.
We note that this assumption can be lifted and the algorithm can be easily restated to be asynchronous.
Indeed, having one agent that is the leader one can simulate synchronous behavior of the agents in such a way that the leader waits for the completion of the current move of another searcher and then informs the searcher that is supposed to perform the next move, dictated by the search strategy, to start the move.

As to the memory model, our algorithm requires that the memory size of the searchers is polynomial in the size of the network, and we do not attempt to optimize this parameter.

\subsection {Partial grid notation} \label{sec:gridnotation}

We assume that a partial grid graph is embedded into two-dimensional Cartesian coordinate system, with a horizontal x-axis and vertical y-axis. For convenience, the homebase is located in a $(0,0)$ position.
We define a partial grid $\grid = (V, E)$ as a set of nodes $V$ and edges $E$, where following conditions hold:
\begin{enumerate}
\item for every ${v = \ver{x}{y} \in V}$, $x$ and $y$ are integer coordinates of the node $v$,
\item for any two adjacent nodes $\ver{x}{y}$ and $\ver{x'}{y'}$ the distance between $(x,y)$ and $(x',y')$ equals one (in Euclidean metric).
\end{enumerate}
In this work $n$ denotes an upper bound of the number of nodes of a partial grid, such that $\sqrt{n}$ is an integer.

Let us notice here that nodes at distance one in the grid are not necessarily neighbors in the graph, thus a partial grid defined in such a way can take all possible shapes, i.e., it can be a tree or a mesh, and it can contain holes.
We note that some simple graphs can be embedded in various ways and for different embeddings our algorithm may perform differently.

Informally speaking, our algorithm will conduct a search by expanding the clean part of the graph from one `checkpoint' to another.
These checkpoints (defined formally later) will be subsets of nodes and their potential placements on the partial grid are dictated by a concept of a \textit{frontier}.
Take any $x = i\size$ for some integer $i$, $y = j\size$ for some integer $j$ and take $i',j' \in \{0,1\}, i' \neq j'$.
Then, the line segment with endpoints $(x,y)$ and $(x + \size i',y + \size j')$ is called a \emph{frontier} and denoted by $\fr{x,y}{x + \size  i',y + \size j'}$.
Whenever the endpoints of a frontier are clear from the context or not important we will omit them. The frontier $\fr{0,0}{\sqrt{n},0}$ that contains the origin is called the \textit{homebase frontier} and the set of all frontiers is denoted by $\aFr$. We will also divide frontiers into \textit{vertical} and \textit{horizontal} ones, where coordinates of two extreme nodes do not differ on first and second coordinate, respectively.

The subgraph induced by all nodes that belong to a frontier $F$ is denoted by $\grid[F]$.

For $i \in \{1,\dots,\size\}$ and some frontier $F=\fr{x,y}{x',y'}$, where $x\leq x'$ and $y\leq y'$, we define \textit{an $i$-th rectangle of} $F$, denoted by $\rectOfF{F}{i}$, as the rectangle with corner vertices $(x-i,y-i)$, $(x-i,y+i)$, $(x'+i,y'-i)$, $(x'+i,y'+i)$ if $F$ is horizontal and as the rectangle with corner vertices $(x-i,y-i)$, $(x+i,y-i)$, $(x'-i,y'+i)$, $(x'+i,y'+i)$ if $F$ is vertical. 

\medskip
Informally speaking, the two above concepts, namely frontiers and rectangles, provide a template on how the search may progress.
However, due to the structure of a partial grid it may be possible that only certain nodes, but not all, that lie on a frontier are reached at some point of a search strategy. For this reason, our notation needs to be extended to subsets of nodes that lie on frontiers and the corresponding rectangles.
Let $F=\fr{x,y}{x',y'}$ be some frontier.
Any subset $C$ of nodes of $\grid$ that belong to $F$ is called a \emph{checkpoint}.
\textit{The $0$-th expansion of a checkpoint $C$} is $C$ itself and is denoted by $\expOfC{C}{0}$.
For $i \in \{1,\dots,\size\}$ we define \textit{an $i$-th expansion of $C$}, denoted by $\expOfC{C}{i}$, recursively as follows:
the set $\expOfC{C}{i}$ consists of all nodes $v\notin\expOfC{C}{0}\cup\expOfC{C}{1}\cup\cdots\cup\expOfC{C}{i-1}$ for which there exists a node $u \in \expOfC{C}{i-1}$, such that there exists a path between $v$ and $u$ in the subgraph of $\grid$ induced by nodes that lie on the rectangles $\rectOfF{F}{0},\rectOfF{F}{1}, \ldots, \rectOfF{F}{i}$.
Define
\[\expOfCplus{C}{i}=\expOfC{C}{0}\cup\ldots\cup\expOfC{C}{i}, \quad i\in\{0,\ldots,\size\}.\]

Informally, $\expOfC{C}{i}$ consists of only those nodes that belong to the rectangle $\rectOfF{F}{i}$ that are connected to nodes of $C$ by paths that lie `inside' of $\rectOfF{F}{i}$ --- this definition captures the behavior of searchers (in our algorithm) that guard the nodes of $C$ and `expand' from $C$ in all directions: then possible nodes that belong to any of the rectangles $\rectOfF{F}{0},\rectOfF{F}{1}, \ldots, \rectOfF{F}{i}$ but do not belong to $\expOfCplus{C}{i}$ will not be reached by the searchers.
See Figure \ref{fig:expansions} for an exemplary checkpoint with its expansions.

\begin{figure}[htb]
\begin{center}
\includegraphics[scale=0.9]{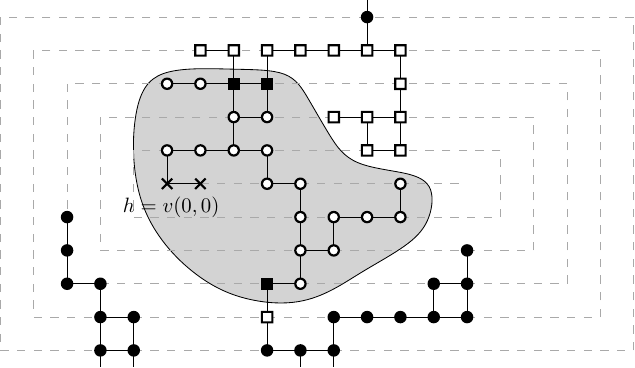}
\caption{Exemplary expansions of a checkpoint $C$ (here $\size=9$); \textit{crosses} denote $C = \expOfC{C}{0}$, \textit{gray area} covers nodes that belong to $\expOfCplus{C}{3}$, \textit{empty squares} denote nodes in $\expOfC{C}{4}$ and \textit{dark squares} denote the one that need to be guarded provided that the gray area consists of the clean nodes.}
\label{fig:expansions}
\end{center}
\end{figure}

\section{Lower bound} \label{sec:lower}

\newcommand{\cL}{\mathcal{L}}

First note that a regular
$\sqrt{n}\times \sqrt{n}$ grid requires $\Omega(\sqrt{n})$ searchers even in the offline setting \cite{EllisW08}, that is, when the network is know in advance and the searchers may decide on the location of the homebase.
Therefore, our distributed algorithm is asymptotically optimal with respect to this worst case measure.

What we would also like to obtain is a lower bound expressed as a competitive ratio, which is defined as a maximized over all networks and all starting nodes ratio between number of searchers that a given algorithm uses and the search number that is optimal for a given network in an offline settings.
In other words, for any distributed algorithm $A$, let $A(G,h)$ be the number of searchers that it uses to clean a network $G$ starting from the homebase $h$ and let $\smash{\displaystyle A(G)=\max_{h}{A(G,h)}}$.
We aim at proving that for \emph{each} distributed algorithm $A$ there exists an $n$-node partial grid network $G$ such that $A(G)/\mcs{G}=\Omega(\sqrt{n}/\log n)$.\footnote{We remark that we define this competitive ratio by taking the worst case homebase for $A$ and in the definition of $\mcs{G}$ the most favorable homebase is selected. However, we note that this does not weaken the result of this section as, informally speaking, one may take two copies of each grid obtained in this section, rotate one copy by 180 degrees and merge the two copies at their homebases. Then, we obtain that for each choice of the homebase any algorithm is forced to use $\Omega(\sqrt{n})$ searchers for some grids since in one copy the search is conducted as in our following analysis.}

Define a class of partial grids
\[\cL=\bigcup_{l\geq 0}\cL_l,\]
where $\cL_l$ for $l\geq 0$ is defined recursively as follows.
We take $\cL_0$ to contain one network that is a single node located at $(0,0)$.
Then, in order to describe how $\cL_{l+1}$ is obtained from $\cL_l$, $l\geq 0$, we introduce an operation of \emph{extending $G\in\cL_l$ at $i$}, for $i\in\{0,\ldots,l\}$.
In this operation, first take $G$ and add $l+2$ new nodes located at coordinates:
\[(0,l+1),(1,l),\ldots,(j,l+1-j),\ldots,(l+1,0).\]
Call these coordinates the \emph{$(l+1)$-th diagonal}.
For each $j\in\{0,\ldots,i\}$ add an edge connecting the nodes $\ver{j}{l-j}$ and $\ver{j}{l-j+1}$, and for each $j\in\{i,\ldots,l\}$ add an edge connecting the nodes $\ver{j}{l-j}$ and $\ver{j+1}{l-j}$.
Then, obtain $\cL_{l+1}$ as follows: initially take $\cL_{l+1}$ to be empty and then for each $G\in\cL_l$ and for each $i\in\{0,\ldots,l\}$, obtain a network $G'$ by extending $G$ at $i$ and add $G'$ to $\cL_{l+1}$. Notice here that a graph constructed this way is not only a partial grid, but also a tree.

Figure~\ref{fig:lower} shows a network that was obtained from the corresponding network in $\cL_7$ by extending it at $6$.
\begin{figure}[htb]
\begin{center}
\includegraphics[scale=1]{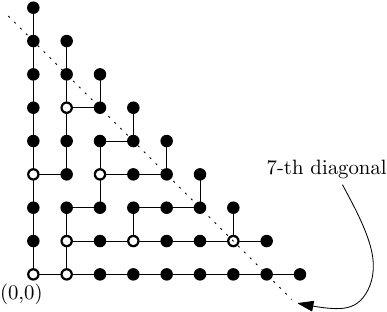}
\caption{A network from $\cL_8$}
\label{fig:lower}
\end{center}
\end{figure}

For a network $G\in\cL_l$, $l\geq 0$, we define a \emph{characteristic sequence of $G$}, $\sigma(G)$, as follows.
If $l=0$, then the characteristic sequence of $G$ is empty.
If $l>0$, then take the network $G'$ such that $G$ has been obtained by extending $G'$ at $i$.
Then, $\sigma(G)=(\sigma(G'),\ver{i}{l-i-1})$ is the characteristic sequence of $G$.
Note that the characteristic sequence uniquely defines the corresponding network.
The network introduced in Figure~\ref{fig:lower} has characteristic sequence $(\ver{0}{0}$, $\ver{1}{0}$, $\ver{1}{1}$, $\ver{0}{3}$, $\ver{3}{1}$, $\ver{2}{3}$, $\ver{1}{5}$, $\ver{6}{1})$.

\begin{lemma} \label{lem:LowerBound}
For any integer $l$ and for each distributed algorithm $A$ computing a connected monotone  search strategy there exists $G\in\cL_{l}$ such that for homebase $\ver{0}{0}$ we have $A(G,\ver{0}{0})\geq (l+1)/2$.
\end{lemma}
\begin{proof}
Consider any algorithm $A$ producing a connected monotone search strategy.
Run $A$ for each network in $\cL_l$ with the homebase $\ver{0}{0}$.
Note that for each network in $\cL_l$, there exist distinct moves $m_1,\ldots,m_l$ such that till the beginning of move $m_j$, $j\in\{1,\ldots,l\}$, no node on the $j$-th diagonal has been occupied by a searcher and at the end of $m_j$ some node $\ver{x_j}{y_j}$ of the $j$-th diagonal is occupied by a searcher.
Consider $G\in\cL_l$ such that $\sigma(G)=(\ver{0}{0},\ver{x_1}{y_1},\ldots,\ver{x_{l-1}}{y_{l-1}})$.
Informally speaking, whenever the algorithm reaches for the first time a node $\ver{i}{j-i}$ in the $j$-th diagonal, an \emph{adversary} decides to extend at $i$ the network explored so far, thus always forcing the situation that the first node reached on a diagonal is of degree three.

Note that at the beginning of move $m_j$, $j\in\{1,\ldots,l\}$, no node of the $j$-th diagonal has been reached by a searcher and the first $j$ nodes of the characteristic sequence have been reached by searchers.
Recall that $G$ is a binary tree.

We analyze the explored part of $\cL_l$ at the beginning of the move $m_l$.
All edges incident to the leaves in $\cL_l$ are contaminated at this point.
On the other hand, all nodes of the characteristic sequence have been visited by searchers till the end of the move $m_l-1$.
Therefore, the contaminated subgraph of $\cL_l$ at this point is a collection of paths leading from nodes that are guarded to the leaves.
Since there are $l+1$ leaves in $\cL_l$, there are $l+1$ such paths, each such path needs to have a searcher placed at one of its endpoints (the one that is not a leaf in $\cL_l$) and, by construction of $\cL_l$, any searcher can be present on at most two such endpoints.
Thus, at least $(l+1)/2$ nodes need to be occupied by searchers, as required by the lemma.
\end{proof}

\begin{theorem} \label{thm:LowerBound}
For each distributed algorithm $A$ computing a connected monotone search strategy there exists an $n$-node network $G$ with homebase $h$ such that
\[\frac{A(G,h)}{\mcs{G}}=\Omega(\sqrt{n}/\log n).\]
\end{theorem}
\begin{proof}
Observe that each network $G$ in $\cL$ is a tree and therefore $\mcs{G}=O(\log(n))$, $n=|V(G)|$ \cite{BFFFNST12,MegiddoHGJP88}.
The theorem follows hence from Lemma~\ref{lem:LowerBound} and the fact that the length of characteristic sequence of each network in $\cL_l$ is $\Omega(\sqrt{n})$.
\end{proof}

\section {The algorithm} \label{sec:algorithm}
In this section we describe our algorithm that takes an upper bound on the size of the network as an input.
Section~\ref{sec:alg-init} deals with the initialization performed at the beginning of the algorithm.
Then, Section~\ref{sec:alg-procedures} introduces two procedures used by the algorithm and finally Section~\ref{sec:alg-main} states the main algorithm.
After each move performed by searchers, each searcher that occupies a node that does not need to be guarded is said to be \emph{free}.
Each node that needs to be guarded is occupied by at least one searcher; if more searchers occupy such node then all of them except for one are also \emph{free}.
If, at some point, no node of the last expansion of some checkpoint need to be guarded, then we say that the expansion is \emph{empty}.

\subsection{Initialization} \label{sec:alg-init}

We start presenting our algorithm by describing initial conditions. Recall that the origin $\ver{0}{0}$ of the two-dimensional xy coordinate system is situated in the homebase. The initial checkpoint $C_0$ is the set of nodes of the connected component of $\grid[F]$ that contains $h$, where $F$ is the homebase frontier. Thus, initially $|C_0|$ searchers place themselves on all nodes of $C_0$ (note that the nodes of $C_0$ induce a path in $\grid$). See Figure \ref{fig:frontier0} for an example.

\begin{figure}[htb]
\begin{center}
\includegraphics[scale=1]{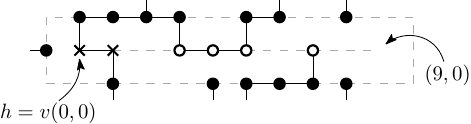}
\caption{Exemplary initialization for $\size = 9$; \textit{crosses} denote nodes belonging to the initial checkpoint $C_0$ and \textit{empty circles} denote nodes that belong to the homebase frontier, but do not fall into $C_0$.}
\label{fig:frontier0}
\end{center}
\end{figure}

\subsection{Procedures} \label{sec:alg-procedures}

\subsubsection{Procedure $\ProcClean$}
We start with an informal description of the procedure.
When a new checkpoint $C$ has been reached, our search strategy `expands' from $C$ by successively cleaning subgraphs $\grid[\expOfCplus{C}{i}]$ for $i\in\{1,\ldots,\size\}$.
Once all nodes in $\expOfCplus{C}{i-1}$ are clean for some $0<i\leq\size$, the transition to reaching the state in which all nodes in $\expOfCplus{C}{i}$ are clean requires cleaning all nodes of the $i$-th expansion of $C$.
This is done by calling for every guarded node $u$ from $\expOfCplus{C}{i-1}$ a special procedure ($\ModProcConnSearch$, described below), which cleans nodes which belong to $\expOfC{C}{i}$ and `has access' to them from $u$.
Procedure $\ProcClean$ performs this job by using $O(\size)$ searchers.

For cleaning all nodes of the $i$-th expansion of $C$, provided that $\grid[\expOfCplus{C}{i-1}]$ is clean we will use a procedure from \cite{BorowieckiDK15} which is more general but for our purposes can be stated as follows.
\begin{theorem}[\cite{BorowieckiDK15}] \label{thm:sense-of-dir}
Let $F$ be any frontier and let $G'$ be any connected grid with nodes lying on rectangles $\rectOfF{F}{0},\rectOfF{F}{1}, \ldots,\rectOfF{F}{i}$, $i\geq 0$.
There exists a distributed procedure $\ProcConnSearch$ that, starting at an arbitrarily chosen homebase in $G'$, clears $G'$ in a connected and monotone way using $6i+4$ searchers.
\qed
\end{theorem}
Note that while using procedure $\ProcConnSearch$, we will be cleaning a subgraph of $\grid[\expOfC{C}{i}]$ that is embedded into the entire partial grid and thus some nodes $v$ of $\grid[\expOfC{C}{i}]$ have edges leading to neighbors that lie outside of $\grid[\expOfC{C}{i}]$.
If such an edge is already clean, then no recontamination happens for the node $v$ and moreover no searcher used by $\ProcConnSearch$ for the subgraph of $\grid[\expOfC{C}{i}]$ needs to stay at $v$.
On the other hand, if such an edge is contaminated (and thus not reached yet by our search strategy), then $v$ needs to be guarded and for that end we place an extra searcher on it that guards $v$ during the remaining execution of $\ProcConnSearch$.
Note that in the latter case, the node $v$ belongs to $\rectOfF{F}{i}$, where $F$ is the frontier that contains the nodes of $C$ and therefore there exist $O(\size)$ such nodes $v$.
In other words, $\ProcConnSearch$ is called to clean a certain subgraph contained within $\rectOfF{F}{i}$ and whenever a node on the rectangle $\rectOfF{F}{i}$ has a contaminated edge leading outside of the rectangle $\rectOfF{F}{i}$, then an extra searcher, no accommodated by $\ProcConnSearch$ in Theorem~\ref{thm:sense-of-dir}, is introduced to be left behind to guard $v$.
The modification of $\ProcConnSearch$ that leaves behind a searcher on each such newly reached node of $\rectOfF{F}{i}$ will be denoted by $\ModProcConnSearch$.
Note that this procedure is invoked for every guarded node from $\expOfCplus{C}{i-1}$ in order to clean $\expOfC{C}{i}$, see Figure \ref{fig:cleandeep} for an example.
\begin{figure}[htb]
\begin{center}
\includegraphics[scale=1]{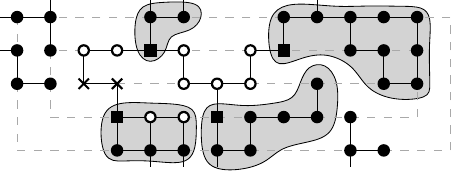}
\caption{Example of an execution of procedure $\ProcClean$; \textit{crosses} denote $C = \expOfC{C}{0}$, \textit{empty circles} denote nodes that belong to $\expOfCplus{C}{1}$, \textit{dark squares} denote the one that belongs to $\expOfCplus{C}{1}$ and for which procedure $\ModProcConnSearch$ is invoked, \textit{gray areas} show nodes that will be cleaned in four calls of $\ModProcConnSearch$ in order to clean $\expOfC{C}{2}$. Note that there are some nodes that belongs to $\expOfCplus{C}{1}$ and are guarded at first, but after one of the calls of $\ModProcConnSearch$ there is no need to guard them any more, so the procedure is not invoked for them.}
\label{fig:cleandeep}
\end{center}
\end{figure}
It follows that it is enough to provide as an input to $\ModProcConnSearch$: a node in $\expOfCplus{C}{i-1}$ that plays the role of homebase for $\ModProcConnSearch$, the frontier $F$ and $i$.
We note that each checkpoint constructed in our final algorithm is obtained as follows: some frontier $F$ is selected and then a checkpoint $C$ is created as some set of nodes that belong to $F$; thus we assume that with $C$ is associated such a unique frontier $F$.

Thus, this approach guarantees us using at most $6i+4$ searchers to clean $\grid[\expOfC{C}{i}]$ and $2\size + 8i$ searchers for guarding nodes laying on $\rectOfF{F}{i}$, which will be analyzed in more details in Section~\ref{sec:proof}. 

To summary, we give a formal statement of our procedure.

  \begin{algorithm}
   \caption{$\ProcClean$}
   \begin{algorithmic}
     \Require An expansion $\expOfC{C}{i-1}$ with $C$ contained in the frontier $F$,  $i\geq 1$.
     \Ensure  Cleaning all nodes of $\expOfC{C}{i}$.
       \While{there exists a node $v\in \expOfC{C}{i-1}$ with a contaminated neighbor $u$ in 					$\expOfC{C}			{i}$}
       \State Place $6i+4$ free searchers on $v$.
        \State Call $\ModProcConnSearch$ for $v$ as the homebase, frontier $F$ and integer $i$.
      \EndWhile
   \end{algorithmic}
  \end{algorithm}

\subsubsection{Procedure $\ProcUpgrade$}

Our algorithm maintains a collection $\checkpoints$ of currently used checkpoints.
Note that $\rectOfF{F}{\size}$, where $F$ is some frontier, contains $10$ frontiers.
Thus, reaching the $\size$-th expansion of a checkpoint of $F$ provides a possibility of creating one new checkpoint for each of the above frontiers.
Procedure $\ProcUpgrade$ generates these new checkpoints, adds them to $\checkpoints$ and removes $C$ from $\checkpoints$.
Also, if it happens that some newly constructed checkpoint belongs to the same frontier as some existing checkpoint in $\checkpoints$ and no expansion for the existing one has been performed yet, then both checkpoints are merged into one.
Finally, any checkpoint in $\checkpoints$, whose lastly performed expansion is empty, is removed from $\checkpoints$.
We remark that only procedure $\ProcUpgrade$ modifies the collection of checkpoints $\checkpoints$ and this procedure performs no cleaning moves.

  \begin{algorithm}
   \caption{$\ProcUpgrade$}
   \begin{algorithmic}
     \Require $\expOfC{C}{\size}$ and the collection of all checkpoints $\aC$
     \Ensure Updated collection $\checkpoints$
     \State  $\checkpoints\leftarrow\checkpoints\setminus\{C\}$
     \State  $\aC_{new}\leftarrow\emptyset$
     \ForAll{frontier $F$ on $\size$-th rectangle of the frontier containing $C$}
        \State Let $C'$ consist of all guarded nodes in $F$.
        \State If $C'\neq\emptyset$, then $\aC_{new}\leftarrow\aC_{new}\cup\{C'\}$.
     \EndFor
     \ForAll {$C''$ in $\checkpoints$}
     	\If {there exists $C'\in\aC_{new}$ that is a subset of the same frontier as $C''$}
           \If {$C''$ is in $0$-th expansion}
            	\State $\aC_{new}\leftarrow\aC_{new}\setminus\{C'\}$
                \State Replace $C''$ with $C''\cup C'$ in $\checkpoints$.
           \EndIf
        \EndIf
      \EndFor
      \State $\aC \leftarrow \aC \cup \aC_{new}$
      \For {each $C$ in $\aC$}
      	\If {no node in the last expansion of $C$ is guarded}
        	\State $\checkpoints \leftarrow \checkpoints\setminus\{C\}$
        \EndIf
      \EndFor
   \end{algorithmic}
  \end{algorithm}

Thus, to summary, the `lifetime' of a checkpoint is as follows.
A newly created checkpoint $C$ may change by getting more nodes only before its $1$-st expansion occurs.
This happens if for another checkpoint its $\size$-th expansion is performed and as a result some new nodes that belong to the same frontier as $C$ become guarded.
Once the $1$-st expansion of $C$ is performed, the checkpoint will remain in the collection $\checkpoints$ and possibly more expansions of $C$ are made (in total at most $\size$ expansion are possible for each checkpoint).
Finally, $C$ may disappear from $\checkpoints$ in two ways: either some expansion of $C$ becomes empty (then $C$ is not removed from $\checkpoints$ right away but during the subsequent call to $\ProcUpgrade$), or $C$ reaches its $\size$-th expansion and procedure $\ProcUpgrade$ is called for $C$ (in which case $C$ possibly `gives birth' to new checkpoints during the execution of $\ProcUpgrade$).

\subsection{Procedure \ProcCS} \label{sec:alg-main}
$\ProcCS$ is the main algorithm, whose aim is to clear the entire partial grid $\grid$ in the connected and monotone way. We start with an informal introduction of the algorithm. The search strategy it produces is divided into phases, which formally will be defined in the next chapter. In each step of the algorithm, the checkpoint with the highest number of nodes that need to be guarded is being chosen and the next expansion is being made on it. When one of the checkpoints reaches its $\size$-th expansion, then the current phase ends and the procedure $\ProcUpgrade$ is being invoked.
Thus, the division of the search strategy into phases is dictated by consecutive calls to procedure $\ProcUpgrade$.
For an expansion $C$, in the pseudocode below we write $\border{C}$ to refer to the set of nodes that belong to the last expansion of $C$ and need to be guarded at a given point.

  \begin{algorithm}
   \caption{$\ProcCS$}
   \label{alg:search}
   \begin{algorithmic}
     \Require An integer $n$ providing an upper bound on the size of the partial grid $\grid$.
     \Ensure A monotone connected search strategy for $\grid$.
     \State Perform the initialization (see Section~\ref{sec:alg-init}).
     \While {$\grid$ is not clean}
     	\While {no checkpoint has reached its $\size$-th expansion}
        	\State Let $C_{max}\in\checkpoints$ be such that $\border{C_{max}}\geq \border{C}$ for each $C\in\checkpoints$.
            \State Let $i$ be the number of expansions of $C_{max}$ performed so far.
        	\State Invoke $\ProcClean$ for $\expOfC{C_{max}}{i}$.
     	\EndWhile
     \State Invoke $\ProcUpgrade$ for $\expOfC{C_{max}}{\size}$ and $\checkpoints$.
     \EndWhile
   \end{algorithmic}
  \end{algorithm}
 
We close this chapter with giving examples of first three expansions of some checkpoint $C$, see Figure \ref{fig:check}, and showing how our algorithm clears an exemplary partial grid network, see Figure \ref{fig:alg}.
  \begin{figure}[htb]
\centering
\begin{subfigure}{.5\textwidth}
  \centering
  \includegraphics[scale=0.82]{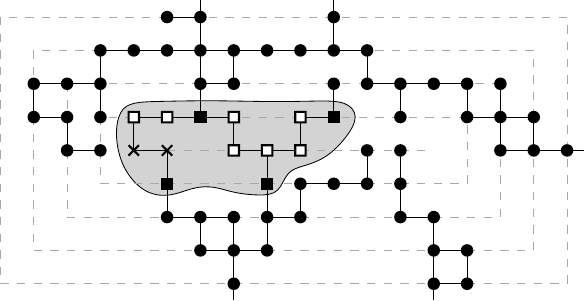}
  \caption{First expansion of $C$.}
  \label{fig:check1}
\end{subfigure}~~
\begin{subfigure}{.5\textwidth}
  \centering
  \includegraphics[scale=0.82]{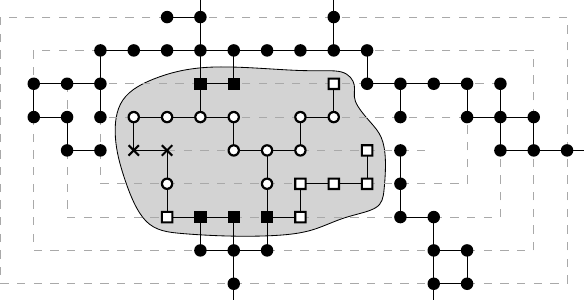}
  \caption{Second expansion of $C$.}
  \label{fig:check2}
\end{subfigure}
\begin{subfigure}{.5\textwidth}
  \centering
  \includegraphics[scale=0.85]{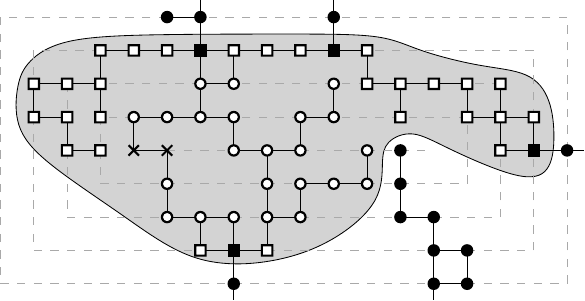}
  \caption{Third expansion of $C$.}
  \label{fig:check3}
\end{subfigure}

\caption{First three expansions for some checkpoint $C$ (here $\size=9$); \textit{crosses} denote $C = \expOfC{C}{0}$, \textit{empty circles} denote nodes cleaned in previous expansions; \textit{squares} denote nodes explored in current expansion; \emph{dark circles} are nodes not reached yet by the searchers; and \textit{dark squares} denote nodes that need to be guarded at the end of current expansion. \textit{Gray areas} show the clean part of the graph, i.e., $\expOfCplus{C}{i}$ for $i\in\{1,2,3\}$.}
\label{fig:check}
\end{figure}

  \begin{figure}
\centering
\begin{subfigure}{.45\textwidth}
  \centering
  \includegraphics[width=0.7\textwidth]{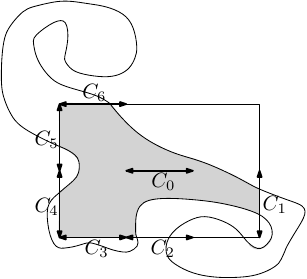}
  \caption{At the end of the first phase $C_0$ (initial checkpoint) reaches its $\size$-th expansion. Procedure $\ProcUpgrade$ creates $6$ new checkpoints and removes $C_0$ from $\checkpoints$, i.e. $\checkpoints = \{ C_1, C_2, C_3, C_4, C_5, C_6\}$.}
  \label{fig:alg1}
\end{subfigure}~~~~
\begin{subfigure}{.45\textwidth}
  \centering
  \includegraphics[width=0.7\textwidth]{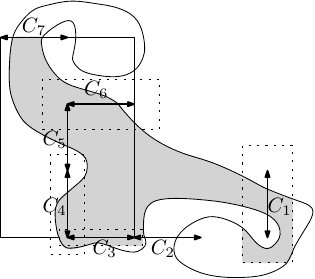}
  \caption{$\size$-th expansion of $C_5$ ends the second phase. Checkpoints $C_4$ and $C_6$ are removed from $\checkpoints$ because (in our example) there is no need to guard any node on theirs expansions; $\checkpoints = \{ C_1, C_2, C_3, C_7\}$.}
  \label{fig:alg2}
\end{subfigure}

\begin{subfigure}{.45\textwidth}
  \centering
  \includegraphics[width=0.7\textwidth]{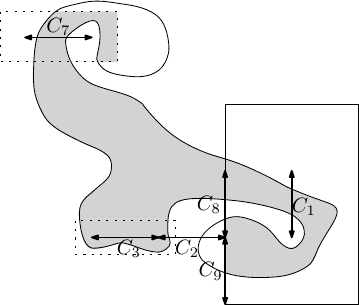}
 \caption{$C_1$ ends the third phase. Notice that a new checkpoint $C_8$ emerged `inside' already cleaned area by $C_0$; $C_2$ is removed from $\checkpoints$ even if $\size$-th expansion has not been reached but its last expansion has no nodes to be guarded; $\checkpoints = \{C_7, C_8, C_9\}$.}
  \label{fig:alg3}
\end{subfigure}~~~~
\begin{subfigure}{.45\textwidth}
  \centering
  \includegraphics[width=0.7\textwidth]{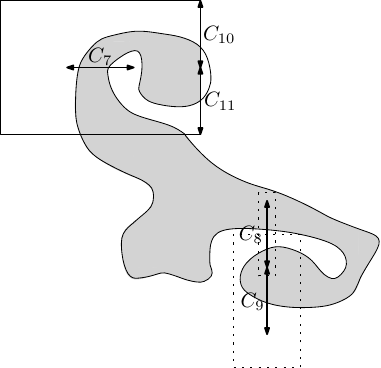}
  \caption{$C_7$ ends the fourth phase. Notice that a new checkpoint $C_{11}$ emerged on an edge of $C_5$'s $\size$-th expansion, but it could not be created in the second phase because then there was no access to the contaminated part; $\checkpoints = \{C_{10}, C_{11}\}$.}
  \label{fig:alg4}
\end{subfigure}

\begin{subfigure}{\textwidth}
  \centering
  \includegraphics[width=0.35\textwidth]{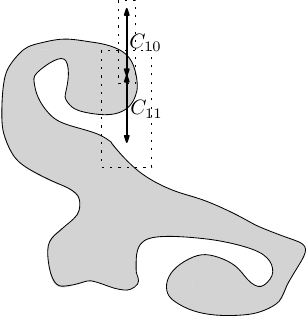}
  \caption{Last phase, in which the rest of the graph is cleaned.}
  \label{fig:alg5}
\end{subfigure}
\caption{Clearing an exemplary partial grid by procedure $\ProcCS$; \textit{gray areas} denote the clean part, \textit{arrows} denote
frontiers on which the marked checkpoints lie, \textit{dotted rectangles} around checkpoints denote their current expansions and \textit{solid rectangles} denote the $\size$-th expansions, which end phases.}
  \label{fig:alg}
 
 \end{figure}

\section {Analysis of the algorithm} \label{sec:proof}

By a \emph{step of the algorithm}, or simply a \emph{step}, we mean all searching moves performed during a single iteration of the internal `while' loop of procedure $\ProcCS$.
Thus, one step of the algorithm includes all moves produced by one call to procedure $\ProcClean$.
A \emph{phase} of an algorithm consists of all its steps between two consecutive calls to procedure $\ProcUpgrade$.
Note that phases may differ with respect to the number of steps they are made of.

We say that a checkpoint is \emph{present} in a given phase if its last expansion is not empty at the beginning of this phase, i.e., if this checkpoint belongs to $\checkpoints$ at the beginning of the phase. Similarly, a checkpoint is \emph{present} in a given step if it is present in the phase to which the step belongs to.
Thus, in particular, a checkpoint is present in none or in all steps of a given phase. Note that some checkpoints may have empty expansions during a part of a the phase, but they still remain present to the end of the phase; this assumption is made to simplify the analysis of the algorithm.

For the purposes of the next definition we say that, for a checkpoint $C$, a node $v$ and a step $t$, the checkpoint $C$ \emph{owns $v$ in step $t$} if:
\begin{itemize}
\item $v$ needs to be guarded at the beginning of step $t$ and $v$ belongs to the last expansion of $C$ performed till the end of step $t-1$, and
\item either this last expansion of $C$ occurred in step $t-1$, or $C$ owns $v$ in step $t-1$.
\end{itemize}
(Intuitively, if a node $v$ is reached by searchers in a step in which an expansion of $C$ occurred, then $C$ owns $v$ as long as $v$ is guarded.)
Given a checkpoint $C$ present in a step $t$, we write $\expansionOfC{C}{t}$ to denote the set of nodes that $C$ owns in step $t$.
The \emph{weight of a checkpoint} $C$ present in a step $t$ is $\weightOfC{C}{t} = |\expansionOfC{C}{t}|$.
Note that each guarded node is owned by exactly one checkpoint and hence, for a step $t$, the sum of weights of all checkpoints present in step $t$ equals the number of nodes that need to be guarded.

If a checkpoint $C$ is not present in a step $t$, then we take $\weightOfC{C}{t} = 0$.
The checkpoint $C_{max}$ selected in a step $t$ (see the pseudocode of Procedure~$\ProcCS$) is called \emph{active in step $t$}, or simply \emph{active} if the step is clear from the context or not important.
All other checkpoints present in this step are called \emph{inactive}.
We define an \emph{active interval} of a checkpoint $C$ to be a maximal interval $[t',t'']$ such that $C$ is active in all steps $t\in\{t',\ldots,t''\}$.

\subsection{Single phase analysis --- how weights of checkpoints evolve}

We now prove lemmas that characterize how the weight of a checkpoint changes over time --- see  Figure~\ref{fig:lifeline} for an exemplary lifeline of a checkpoint.
Informally, the weight of a checkpoint $C$ does not grow in intervals in which $C$ is inactive (Lemma~\ref{lem:inactive_notgrow}).
Also, the weight of $C$ at the end of an active interval is not greater than at the beginning of it
(Lemma~\ref{lem:active_notgrow}); however, no upper bounds except for the trivial one of $O(\size)$ can be concluded for the weight of $C$ inside its active interval.
\begin{figure}[htb]
\begin{center}
\includegraphics[width=0.9\textwidth]{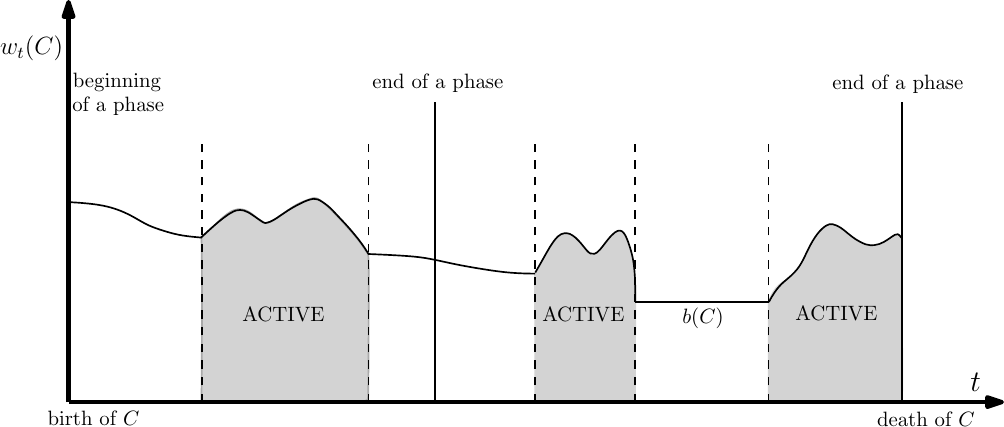}
\caption{Exemplary lifeline of a checkpoint $C$.}
\label{fig:lifeline}
\end{center}
\end{figure}

\begin{lemma}\label{lem:inactive_notgrow}
If a checkpoint $C$ is present and inactive in a step $t$, then $\weightOfC{C}{t+1} \leq \weightOfC{C}{t}$.
\end{lemma}
\begin{proof}
It follows directly from the definitions and procedure $\ProcClean$ that the only checkpoint on which an expansion is performed during execution of $\ProcClean$ is the active one. The weight of an inactive checkpoint $C$ can change only in the situation, where the active checkpoint in a step $t$ expands on some nodes that belong to the last expansion of $C$. In other words, the weight of $C$ may decrease if $C$ contains in step $t$ nodes that are added to the active checkpoint in step $t+1$.
Thus, if $t$ is not the last step of a phase, then the proof is completed.

If $t$ is the last step of some phase, then apart from procedure $\ProcClean$, procedure $\ProcUpgrade$ is being invoked, which affects $C$ in two situations:
\begin{itemize}
\item there exists a step $t'$ in the phase that ends such that $\weightOfC{C}{t'} = 0$. Then, because $C$ can not be expanded during steps $t',\ldots,t$ of the phase, we get directly that $\weightOfC{C}{t + 1} = \weightOfC{C}{t} = 0$. 
\item $C$ is in its $0$-th expansion and a new checkpoint is placed on the same frontier, which implies that $C$ is not present in step $t+1$ and thus $\weightOfC{C}{t + 1} = 0$.
\end{itemize}
Thus, in all cases we obtain that $\weightOfC{C}{t+1} \leq \weightOfC{C}{t}$.
\end{proof}

We next observe that, informally speaking, once a checkpoint becomes active, it remains active until either the phase ends or its weight decreases. Note that a checkpoint that is active in the last step of the phase is not present in the first step of the next phase, i.e. its weight is then zero, which allows us to state the lemma as follows:
\begin{lemma}\label{lem:active_notgrow}
Let $C$ be a checkpoint and let $[t',t'']$ be an active interval of $C$. Then, $\weightOfC{C}{t''+1}\leq\weightOfC{C}{t'}$.
\end{lemma}
\begin{proof}
Obviously, $t'$ and $t''$ must belong to the same phase, because at the end of each phase the active checkpoint is removed from $\checkpoints$, i.e., it is no longer present in the next phase.

If $t''$ is the last step of the phase then the lemma follows, because $\weightOfC{C}{t''+1} = 0\leq\weightOfC{C}{t'}$.

We will now prove that lemma holds when $t''$ is not the last step of the phase. Let us suppose for a contradiction that $\weightOfC{C}{t''+1} > \weightOfC{C}{t'}$. From the assumptions of the lemma and definition of an active interval we get that $C$ is not the active checkpoint in step $t''+1$. Because we are still in the same phase, it means that there must exist a checkpoint $C^*$ such that $\weightOfC{C^*}{t''+1} \geq \weightOfC{C}{t''+1}$. Moreover from Lemma \ref{lem:inactive_notgrow} we know, that because $C^*$ was inactive from step $t'$ to $t''$, it holds $\weightOfC{C^*}{t'} \geq \weightOfC{C^*}{t''} \geq \weightOfC{C^*}{t'' + 1}$. This gives us 
\[\weightOfC{C^*}{t'} \geq \weightOfC{C^*}{t''+1} \geq \weightOfC{C}{t''+1} > \weightOfC{C}{t'},\]
which is in a contradiction to the assumption, that $C$ is the active checkpoint in step $t'$.
\end{proof}

\begin{remark}\label{rem:active_notgrow}
Let $C$ be a checkpoint and let $[t',t'']$ be an active interval of $C$.
For every step $t \in \{t',\ldots, t''\}$ it holds $\weightOfC{C}{t}\geq\weightOfC{C}{t''+1}$. 
\qed
\end{remark}

We now conclude from the two previous lemmas that the weight of each checkpoint is not greater at the end of a phase than at the beginning of the phase.
\begin{lemma}\label{lem:f_notgrow}
Suppose that a phase starts in step $t'$ and ends in step $t''$.
For each checkpoint $C$ present in this phase it holds $\weightOfC{C}{t''+1}\leq\weightOfC{C}{t'}$.
\end{lemma}
\begin{proof}
Each checkpoint $C$ can be active or inactive in different steps during the whole phase. If in some step $t \in \{t', \ldots, t'' \}$ a checkpoint $C$ is inactive then from Lemma \ref{lem:inactive_notgrow} we have that the weight of it will not increase, i.e., $\weightOfC{C}{t} \geq \weightOfC{C}{t + 1}$. On the other hand, Lemma \ref{lem:active_notgrow} guarantees us, that the weight of an active checkpoint can not be greater at the end of any its active interval than at the beginning.
Since end of a phase is the end of some active interval, this finishes the proof.
\end{proof}

\begin{remark}\label{rem:f_atends}
Suppose that a phase ends in a step $t$ and the next one ends in a step $t'$. If a checkpoint $C$ is inactive (but present) in steps $t$ and $t'$, then $\weightOfC{C}{t'} \leq \weightOfC{C}{t}$.
\qed
\end{remark}

\subsection{How many nodes are explored by a checkpoint?}

Define a \emph{bottleneck} of a checkpoint $C$, denoted by $\bottleneckOfC{C}$ to be its minimum weight taken over all steps in which $C$ was present. (Note that a checkpoint may be present in many consecutive phases.)

Suppose that a node $v$ has been reached by a searcher for the first time in a step $t$.
Let $C$ be the active checkpoint in step $t$.
We say that $v$ has been \emph{explored by} $C$.

If an expansion of an active checkpoint $C$ reaches in a step $t$ a node $u$ already explored by some checkpoint $C'$, then in most situations $u$ does not need to be guarded. However there might occur a ``corner situation'', when $u$ still needs to be guarded in order to avoid contamination. In such case, the algorithm clearly needs one searcher on $u$ to guard it and so it is counted in our analysis due to the `ownership' relation used in the definition of the weight of a checkpoint.

The next lemma states a lower bound on the number of nodes explored by a checkpoint reaching its last expansion.

\begin{lemma}\label{lem:active_search}
Suppose that a phase ends in a step $t$.
Let $C$ be the active checkpoint in step $t$.
The number of nodes explored by $C$ in all steps is at least $\bottleneckOfC{C}\size$.
\end{lemma}
\begin{proof}
First let us make a remark that nodes can be only explored by $C$ during execution of procedure $\ProcClean$ that took $C$ as an input, i.e., when $C$ is active. Let us denote by $S$ the set of all nodes explored by $C$.

Because $C$ is active in the last step of the phase, it had to be active in exactly $\size$ steps in total, which can be concluded in several past phases. Let $t_1, t_2,\ldots, t_{\size} = t$ be all steps in which $C$ is active. 
Note that
\[\bigcup_{i=1}^{\size}\expansionOfC{C}{t_i} \subseteq S\]
and $\expansionOfC{C}{t_i}\cap\expansionOfC{C}{t_j}=\emptyset$ for $i\neq j$.
The latter follows directly from the fact that nodes in $\expansionOfC{C}{t_i}$ and $\expansionOfC{C}{t_j}$ belong to different rectangles of the frontier containing $C$ for $i\neq j$.
(Recall that $|\expansionOfC{C}{t}|=\weightOfC{C}{t}$ for each step $t$.)
Also from definition of the bottleneck, we get that $\bottleneckOfC{C}\leq\weightOfC{C}{t_i}$ for each $i\in\{1,\ldots,\size\}$ and hence we conclude that:
\[|S| \geq \sum_{i=1}^{\size}\weightOfC{C}{t_i} \geq  \bottleneckOfC{C}\size.\]
\end{proof}

We now give an upper bound on the weight of each inactive checkpoint at the end of a phase.
\begin{lemma}\label{lem:inactive_lower}
Suppose that a phase ends in a step $t$.
Let $C_1,\ldots,C_l$ be all checkpoints present in this phase, where $C_1$ is the active checkpoint in step $t$.
Then, $\bottleneckOfC{C_1}\geq\weightOfC{C_j}{t}$ for each $j\in\{2,\ldots,l\}$.
\end{lemma}
\begin{proof}
Let us denote by $t'$ the last step in which $\weightOfC{C_1}{t'} = \bottleneckOfC{C_1}$. If $t'=t$ then the lemma follows strictly from the definition of an active checkpoint.
We will now prove that lemma stands also when $t'<t$.

Suppose that $t'$ and $t$ do not belong to the same active interval of $C_1$.
From the Remark \ref{rem:active_notgrow} we know that $\weightOfC{C_1}{t''} = \bottleneckOfC{C_1}$ occurs for some $t''$ that does not belong to an active interval. Moreover from Lemma \ref{lem:active_notgrow} we get that every next active interval will need to start and finish on the same weight as the bottleneck, which is in contradiction that $t'$ is the last step when $\bottleneckOfC{C_1}$ occurred.

Hence there must exist an active interval of $C_1$ that contains both $t'$ and $t$. Then, we get from Lemma \ref{lem:inactive_notgrow} and the fact that $C_1$ is active in step $t'$:
\[\weightOfC{C_j}{t}   \leq  \weightOfC{C_j}{t'} \leq  \weightOfC{C_1}{t'} = \bottleneckOfC{C_1}, \quad j \in \{2,\ldots,l\},\]
which finishes our proof.
\end{proof}

Let us introduce a relation $\leqF$ on a set of checkpoints.
Whenever $C \leqF C'$, we say that $C$ is a \emph{predecessor} of $C'$ and $C'$ is a \emph{successor} of $C$.
We stress out that the construction depends on the execution of the algorithm, namely only checkpoints that appear in some step are considered, and the division of the steps into phases shapes the relation.
More precisely, the relation is defined only for checkpoints added to the set $\checkpoints$ during all executions of procedure $\ProcUpgrade$.
To construct the relation we iterate over the consecutive phases of the algorithm.
Initially the relation is empty and once the construction is done for each phase smaller than $i$, we perform the following for the phase $i$.
Let $C$ be the active checkpoint in the last step of phase $i$.
Let $C_1,\ldots,C_l$ be all checkpoints, different than $C$, that have no successors so far and were added to $\checkpoints$ till the end of phase $i-1$.
Then, let $C_j \leqF C$ for each $j\in\{1,\ldots,l\}$.

An important property of our algorithm is that each checkpoint may have only constant number of predecessors:
\begin{lemma} \label{lem:number_predec}
Each checkpoint has at most $\nrPred$ predecessors.
\end{lemma}
\begin{proof}
A checkpoint $C$ can only once be active in the last step of some phase $i$, because after that it will not be present in any later phases. At the end of phase $i$ the only checkpoints that do not have any successors are the one that were constructed by the procedure $\ProcUpgrade$ at the end of phase $i-1$. There are at most $\nrPred$ such checkpoints.
\end{proof}

\subsection{The algorithm uses $O(\size)$ searchers in total} \label{sec:totalbound}

We now bound the total weight of all checkpoints at the end of each phase --- note that this bounds the total number of searchers used for guarding at the end of a phase.
A high level intuition behind the proof of Lemma~\ref{lem:nr_guards_end} is as follows.
Since, due to Lemma~\ref{lem:active_search}, each checkpoint $C$ that is active in the last step of a phase explores at least $\bottleneckOfC{C}\size$ nodes in total.
Therefore, the sum of bottlenecks of all such checkpoints $C$ cannot exceed $\size$.
Moreover, $C$ can have at most $10$ predecessors and hence the sum of weights of those predecessors is bounded by $10\bottleneckOfC{C}$ according to Lemma~\ref{lem:inactive_lower}.
Since each checkpoint (except the one that is active in the last step of a given phase) is a predecessor of some checkpoint that is active in the last step of some phase, we bound the sum of all weights of all such checkpoints present in a given phase by $10\size$.

\begin{lemma} \label{lem:nr_guards_end}
Suppose that $C_1,\ldots,C_l$ are all checkpoints of a phase that ends in step $t$, where $C_1$ is active in step $t$.
Then,
\[\sum_{i=1}^l \weightOfC{C_i}{t}  \leq  \weightOfC{C_1}{t} + \nrPred\size.\]
\end{lemma}
\begin{proof}
Suppose that a phase $j$ ends in step $t$.
Let $t_{i}$ be the last step of phase $i$ and let $C_{i}^0$ be the active checkpoint in step $t_{i}$ for each $i\in\{0,\ldots,j\}$. We denote by $s$ the number of nodes visited by searchers till the end of step $t$. From Lemma \ref{lem:active_search} and the fact that the number of all nodes $n$ is at least $s$ we have:
\begin{equation} \label{eq:10b2}
n \geq s \geq \sum_{i=0}^j \bottleneckOfC{C_i^0} \size  \quad \Rightarrow \quad \nrPred\size \geq \nrPred \sum_{i=0}^j \bottleneckOfC{C_i^0}.
\end{equation}

From Lemma \ref{lem:number_predec} we have that the checkpoints $C_0^0,\ldots,C_j^0$ can have at most $\nrPred$ predecessors. From the definition, they are constructed (i.e., added to collection $\checkpoints$ during execution of procedure $\ProcUpgrade$) at the beginning of the first step of a phase at the end of which their successor is active. Let us denote by $C_i^1,\ldots,C_i^{l_i}$, $0 \leq l_i\leq 10$, the predecessors of $C_i^0$ for each $i\in\{0,\ldots,j\}$ (by $l_i = 0$ we understand that $C_i^0$ has no predecessors).
From Lemma \ref{lem:inactive_lower} we have:
\begin{equation} \label{eq:10b1}
\sum_{k=1}^{l_i} \weightOfC{C_i^k}{t_i} \leq \nrPred \bottleneckOfC{C_i^0}, \quad  i\in\{0,\ldots,j\}.
\end{equation}

Remark \ref{rem:f_atends} assures us that weights of inactive checkpoints will not be greater at the end of the next phase than they are in the last step of current phase:
\begin{equation} \label{eq:10Cl}
\weightOfC{C_i^k}{t} = \weightOfC{C_i^k}{t_j} \leq \weightOfC{C_i^k}{t_{j - 1}} \leq \cdots \leq \weightOfC{C_i^k}{t_i}, \quad  i\in\{0,\ldots,j\}; \quad k\in\{1,\ldots,l_i\}.
\end{equation}
Because
\[\{C_1,\ldots,C_l\}\subseteq\{C_j^0\} \cup \left\{C_i^k\st  k\in\{1,\ldots, l_i\}, i\in\{0,\ldots,j\} \right\},\] we can conclude from Equations~\eqref{eq:10Cl}, \eqref{eq:10b1} and \eqref{eq:10b2} (in this order) that:
\begin{eqnarray*} \nonumber
\sum_{i=1}^l \weightOfC{C_i}{t} &\leq& \weightOfC{C_j^0}{t} + \sum_{i=0}^j \sum_{k=1}^{l_i} \weightOfC{C_i^k}{t}\\ \nonumber
&\leq& \weightOfC{C_j^0}{t} + \sum_{i=0}^j \sum_{k=1}^{l_i} \weightOfC{C_i^k}{t_i}\\ \nonumber
& \leq & \weightOfC{C_j^0}{t} +  \sum_{i=0}^j \nrPred \bottleneckOfC{C_i^0} \\ \nonumber
& \leq & \weightOfC{C_j^0}{t} + \nrPred\size.
\end{eqnarray*}
\end{proof}

\begin{theorem} \label{thm:known-n}
Given an upper bound $n$ of the size of the network as an input, the algorithm $\ProcCS$ clears in a connected and monotone way any unknown underlying partial grid network using $O(\sqrt{n})$ searchers.
\end{theorem}
\begin{proof}
We will bound the number of searchers $s$ used by a single call to procedure $\ProcClean$ and the total number of searchers $s'$ used for guarding at the end of each step of the algorithm.
Note that $s+s'$ bounds the total number of searchers used by $\ProcCS$. 

We first analyze procedure $\ProcClean$ to give an upper bound on $s$.
We can divide searchers into three groups: \textit{explorers}, \textit{cleaners} and \textit{guards}.
Suppose that procedure $\ProcClean$ performs $i$-th expansion of a checkpoint $C_{max}$.
Denote by $F_{max}$ the frontier that contains the nodes in $C_{max}$.
All searchers located at nodes on the $(i-1)$-th rectangle of $F_{max}$ that need to be occupied in order to avoid recontamination at the beginning of the call to procedure $\ProcClean$ are named to be guards.
The explorers and cleaners are used by algorithm $\ModProcConnSearch$ called during the execution of procedure $\ProcClean$.
Each time $\ModProcConnSearch$ reaches a node $v$ on the $i$-th rectangle of $F_{max}$ such that $v$ needs to be guarded, the searcher used for guarding $v$ is called an explorer.
The searchers used in $\ModProcConnSearch$ that mimic the movements of searchers in algorithm $\ProcConnSearch$ are the cleaners.
We point out that we do not alter here the behavior of $\ProcClean$ and $\ModProcConnSearch$ but just assign one of the three categories to each searcher they use.
Informally speaking, when explorers protect nodes lying on the $i$-th rectangle and the guards protect the ones lying on the $(i-1)$-th rectangle of $F_{max}$, cleaners clean nodes inside the $i$-th rectangle of $F_{max}$ (i.e., the remaining nodes of the $i$-th expansion of $C_{max}$).

The fact that each rectangle of a frontier contains at most $\nrPred\size$ nodes and Theorem~\ref{thm:sense-of-dir} give that:   
\begin{eqnarray} \nonumber
\text{number of explorers} &\leq& \nrPred \size, \\  \nonumber
\text{number of cleaners} &\leq& 6\size + 4. \\  \nonumber
\end{eqnarray}
Thus,
\[s\leq 16\size +4.\]
The guards used to protect nodes lying on the $(i-1)$-th rectangle are accounted for during the estimation of $s'$ below.

We now bound the maximal number of searchers used for guarding at the end of each step $t$ of our search strategy, which we denote by $g_t$. 
It is easy to see that $g_t \leq 10\size$ if $t$ belongs to phase $0$.

Let us now take any step $t$ that belongs to an $i$-th phase, where $i > 0$ and denote by $t'$ the last step of the phase $i - 1$ and by $C$
the active checkpoint in step $t'$.
From Lemma~\ref{lem:nr_guards_end} we know that $g_{t'} \leq \weightOfC{C}{t'} + \nrPred\size \leq 20\size$. 
The latter inequality follows from the fact that all nodes in $\expansionOfC{C}{t'}$ 
belong to the $j$-th rectangle of the frontier that contains $C$, $j\leq\size$, and the number of nodes in this rectangle 
is at most $\nrPred\size$. 

We know now that every phase starts with at most $20\size$ guards. 
If $t$ is the first step of an active interval of some checkpoint, then by Lemmas \ref{lem:inactive_notgrow} and \ref{lem:active_notgrow}
we have that $g_t \leq g_{t'} \leq 20\size$. 
But if $t$ is a step inside some active interval, then an active checkpoint 
can reach at most $10\size$ new nodes that need to be guarded. Because in one step only one checkpoint can be active that leads us to conclusion that
for every step $t$ we have $g_t \leq 30\size$.
Therefore, we obtain that $s'\leq 30\size$.

Thus, we obtain $s+s'\leq 46\size + 4=O(\size)$ as required. 
\end{proof}

\section {Unknown size of the graph} \label{sec:unknown}

The algorithm we have described needs to know an upper bound of the size of the underlying partial grid network $\grid$.
In this section we design a procedure called $\ModProcCS$ that performs the search using $O(\size)$ searchers and having no prior information on the network.
The procedure is based on a standard technique: guessing an upper bound on $n$ by doubling potential estimate each time.

The procedure $\ModProcCS$ is composed of a certain number of $m$ \emph{rounds}.
In round $i$, procedure $\ProcCS$ first introduces $c\sqrt{2^i}$ new searchers called \emph{$i$-th team}, where $c$ is a constant from the asymptotic notation in Theorem~\ref{thm:known-n}.
Then, a call to $\ProcCS$ is made, where procedure $\ProcCS$ is using only the searchers of the $i$-th team.
The outcome can be twofold.
The procedure may succeed in searching the entire graph and in such case the $i$-th round is the last one and $\ModProcCS$ is completed, or the procedure may encounter a situation in which it would be forced to use more than $c\sqrt{2^i}$ searchers to continue. In such case $\ProcCS$ stops, the $i$-th round ends and the $(i+1)$-th round will follow.
Once the $i$-th round is completed, the searchers of the $i$-th team stay idle indefinitely.
We point out that during the execution of an $i$-th round, $i>1$, procedure $\ProcCS$ using the searchers of the $i$-th team is ignoring the fact that the network may be partially clear as a result of the work done in previous rounds. Moreover, the searchers of $j$-th team for each $j<i$ are not used and thus also ignored during $i$-th round.

We close this section by giving an upper bound of the number of searchers that need to be used in the presented modified version of our algorithm.
\begin{theorem}\label{thm:unknown}
There exists a distributed algorithm that clears (starting at an arbitrary homebase) in a connected and monotone way any unknown underlying partial grid network using $O(\sqrt{n})$ searchers. The algorithm receives no prior information on the network.
\end{theorem}
\begin{proof}
Let $n$ be the number of nodes of  the partial grid network, which is unknown to our procedure.
The number of rounds $m$ fulfills $2^{m-1} < n \leq 2^{m}$, i.e. $m = \ceil{\log_2{n}}$.
At the end of $i$-th round, $c \sqrt{2^i}$ searchers need to stay in their last positions till the end of our procedure and are not used in subsequent rounds.
This means that the total number of searcher $s$ is upper bounded by a sum of searchers used in every round:
\begin{eqnarray*}
s & \leq & c\sqrt{2} + c\sqrt{2^2} + \dots +  c\sqrt{2^{\ceil{\log_2{n}}}} = c\sum_{j=1}^{\ceil{\log_2{n}}}\left( \sqrt{2} \right)^j \\
  & =    & \sqrt{2}c \frac{ 1- \sqrt{2}^{\ceil{\log_2{n}} } }{1 - \sqrt{2} } = \frac{\sqrt{2}c}{\sqrt{2} - 1} \left( \sqrt{2^{\ceil{\log_2{n}} }} - 1 \right).
\end{eqnarray*}
Because $\sqrt{n} \leq \sqrt{2^{\ceil{\log_2{n}}}} < \sqrt{2n}$, we conclude
\[s <  \frac{\sqrt{2}c}{\sqrt{2} - 1} \left(\sqrt{2n} - 1 \right)\quad \Rightarrow \quad s= O(\sqrt{n}) .\]
\end{proof}

\section{Conclusions} \label{sec:conclusions}

\subsection{Graph-theoretic modeling} \label{sec:polygon}
In this section we consider a searching scenario for two-dimensional environment and its modeling via graph theory.
This provides another motivation for the distributed graph searching model used in this work.
Consider a continuous search problem in which $k$ searchers initially placed at the same location need to capture the fugitive hiding in an arbitrary polygon that possibly has holes.
The polygon is not known a priori to the searchers.
The fugitive is considered captured in time $t$ when it is located at distance at most $r$ from some searcher at time point $t$.
(The $r$ can be related to physical dimensions of searchers and/or their visibility range, etc.)

\medskip
Consider the following transition from the above continuous searching problem of a polygon to a discrete one.
Overlap the coordinate system with the polygon in such a way that the origin coincides with the original placement of the searchers.
Then, place nodes on all points with coordinates being multiples of $r$ and lying in the polygon.
Connect two nodes with an edge if the edge is contained in the polygon.
In this way we obtain a partial grid network.
In this brief sketch we omit potential problems that may arise in such modeling, like obtaining disconnected networks or having `blind spots', i.e., points in the polygon that cannot be cleared by using the above nodes and edges only.
We say that a partial grid network $G$ \emph{covers} the polygon if $G$ is connected and for each point $p$ in the polygon there exist a node of $G$ in distance at most $r$ from $p$.
\begin{figure}[htb]
\begin{center}
\includegraphics[scale=0.8]{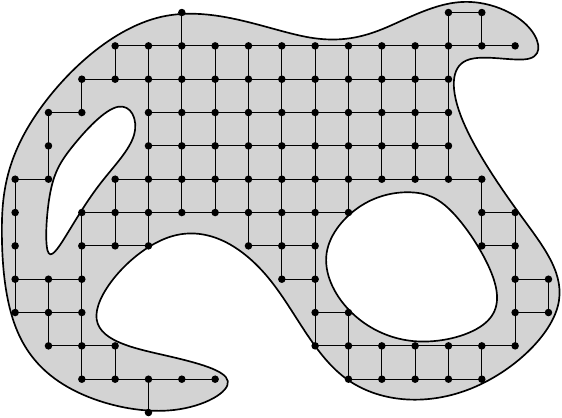}
\caption{An example of the construction of a partial grid network.}
\label{fig:polygon}
\end{center}
\end{figure}
See Figure~\ref{fig:polygon} for an example.

Note that any search strategy $\cS'$ for a polygon $P$ can be used to obtain a search strategy $\cS$ for underlying partial grid network $G$ as follows.
For each searcher $s$ used in $\cS'$ introduce four searchers $s_1,\ldots,s_4$ that will `mimic' its movements by going along edges of $G$.
More precisely, the searchers $s_1,\ldots,s_4$ will ensure that at any point, if $s$ is located at a point $(x,y)$, then $s_1,\ldots,s_4$ will reside on nodes with coordinates $(\lfloor x/r\rfloor,\lfloor y/r\rfloor)$, $(\lfloor x/r\rfloor,\lceil y/r\rceil)$, $(\lceil x/r\rceil,\lfloor y/r\rfloor)$, $(\lceil x/r\rceil,\lceil y/r\rceil)$.
In this way, at any time point, the four searchers in $\cS$ protect an area that contains the area protected by $s$ in $\cS'$.
This allows us to state the following.
\begin{observation}
Let $P$ be a polygon and let $G$ by an underlying partial grid network that covers $P$.
Then, there exists a search strategy for $G$ using $k$ searchers such that its execution in $G$ results in clearing $P$ and $k=O(p)$, where $p$ is the minimum number of searchers required for clearing $P$ (in continuous way).
\end{observation}

\subsection{Open problems}
In view of the lower bound shown in \cite{INS09} that even in such simple networks as trees each distributed algorithm may be forced to use $\Omega(n/\log n)$ times more searchers than the connected search number of the underlying network, one possible line of research is to restrict attention to specific topologies that allow to obtain algorithms with good provable upper bounds. This work gives one such example. An interesting research direction is to find other non-trivial settings in which distributed search can be conducted efficiently.

\medskip
The above questions related to network topologies can be stated more generally: what properties of the distributed model are crucial for such search for fast and invisible fugitive to be efficient?
This work and also a recent one \cite{BorowieckiDK15} suggest that a `sense of direction' may be one such factor.
Possibly interesting directions may be to analyze the influence of visibility on search scenarios.

\medskip
We finally note that the only optimization criterion that was of interest in this work is the number of searchers.
This coincides with the research done in offline search problems where this was the most important criterion giving nice ties between graph searching theory and structural graph theory.
However, one may consider adding different optimization criteria like time (defined as the maximum number of synchronized steps) or cost (the total number of moves performed by all searchers).

\medskip
An interesing research direction is to study agent guaranteed searching algorithms, in this or more general settings, under weaker assumptions regarding agent capabilities like memory size or ability to distinguish directions etc.

\bibliographystyle{plain}

\bibliography{bibliography}

\end{document}